\documentclass[twoside,leqno,twocolumn]{article}
\usepackage{ltexpprt}
\usepackage[keeplastbox]{flushend}
 
\newcommand{\lv}[1]{#1}
\newcommand{\sv}[1]{}

\usepackage{url}
\usepackage{float}
\usepackage{etex}
\usepackage{booktabs}
\usepackage{tabularx}
\usepackage{enumerate}
\usepackage{amsmath}
\usepackage{multirow}
\usepackage{amsfonts}
\usepackage{amssymb}
\usepackage{tikz}
\usepackage{fancybox}
\usepackage{mathtools} 
\usepackage{enumitem}
\usepackage{adjustbox}
\usepackage{pifont}
\usepackage{times}
\usepackage{tikz-qtree}
\usepackage{wrapfig}
\usepackage{xspace}

\usetikzlibrary{tikzmark}
\usetikzlibrary{automata,positioning}
\usetikzlibrary{patterns}
\usetikzlibrary{decorations.text}
\usetikzlibrary{arrows,shapes}
\usetikzlibrary{calc} 
\newcommand{\todo}[1]{}
\renewcommand{\todo}[1]{{\color{red} TODO: {#1}}}
\newcolumntype{R}[1]{>{\raggedleft\arraybackslash}p{#1}}

\newcommand{\PPP}{\mathcal{P}}

\newcommand{\der}{\PPP}

\newcommand{\mheight}{\textup{h}}

\newcommand{\SB}{\{\,}%
\newcommand{\SM}{\mid}
\newcommand{\SE}{\,\}}%

\newcommand{\bigO}{\mathcal{O}}

\newcommand{\adhesion}{\mathsf{ad}}
\newcommand{\torsow}{\mathsf{tor}}
\newcommand{\tcw}{\mathsf{tcw}}
\newcommand{\td}{\mathsf{td}}
\newcommand{\pw}{\mathsf{pw}}
\newcommand{\tw}{\mathsf{tw}}

\newcommand{\comp}{\mathsf{s}}
\newcommand{\led}{\mathsf{l}}

\newcommand{\ctr}{\mathsf{\#}}
\newcommand{\anc}{\mathsf{anc}}
\newcommand{\pt}{\mathsf{par}}
\newcommand{\ro}{\mathsf{ro}}
\newcommand{\qed}{{\hfill\ensuremath{\blacksquare}}}

\newcommand\blfootnote[1]{%
  \begingroup
  \renewcommand\thefootnote{}\footnote{#1}%
  \addtocounter{footnote}{-1}%
  \endgroup
}

\begin{document}
{

\title{\Large SAT-Encodings for Treecut Width and Treedepth}

\author{Robert Ganian\thanks{Algorithms and Complexity Group, TU Wien, Vienna, Austria}
\and Neha Lodha$^*$
\and Sebastian Ordyniak\thanks{Algorithms Group, University of Sheffield, Sheffield, UK}
\and Stefan Szeider$^*$
}

\date{}



\maketitle

 

\begin{abstract}
  The decomposition of graphs is a prominent algorithmic task with numerous applications in computer science. A graph decomposition method is typically associated with a width parameter (such as treewidth) that indicates how well the given graph can be decomposed. Many hard (even \#P-hard) algorithmic problems can be solved efficiently if a decomposition of small width is provided;  the runtime, however, typically depends exponentially on the decomposition width. Finding an optimal decomposition is itself an NP-hard task. In this paper we propose, implement, and test the first practical decomposition algorithms for the width parameters \emph{treecut width} and \emph{treedepth}. These two parameters have recently gained a lot of attention in the theoretical research community as they offer the algorithmic advantage over treewidth by supporting so-called \emph{fixed-parameter algorithms} for certain problems that are not fixed-parameter tractable with respect to treewidth. However, the existing research has mostly been theoretical. A main obstacle for any practical or experimental use of these two width parameters is the lack of any practical or implemented algorithm for actually computing the associated decompositions. We address this obstacle by providing the first practical decomposition algorithms.
  \blfootnote{The authors acknowledge support by the Austrian Science
    Fund (FWF, projects W1255-N23, P31336, and P27721).}

Our approach for computing treecut width and treedepth decompositions is based on efficient encodings of these decomposition methods to the propositional satisfiability problem (SAT). Once an encoding is generated, any satisfiability solver can be used to find the decomposition. This allows us to leverage the surprising power of todays state-of-the art SAT solvers. The success of SAT-based decomposition methods crucially depends on the used  characterisation of the decomposition method, as not every characterisation is suitable for that task. For instance, the successful leading SAT encoding for treewidth is based on a characterisation of treewidth in terms of elimination orderings. For treecut width and treedepth, however, we propose new characterisations that are based on sequences of partitions of the vertex set, a method that was pioneered for clique-width. We implemented and systematically tested our encodings on various benchmark instances, including famous named graphs and random graphs of various density. It turned out that for the considered width parameters, our partition-based SAT encoding even outperforms the best existing SAT encoding for treewidth.

We hope that our encodings---which we will make publicly available---will stimulate the experimental research on the algorithmic use of treecut width and tree depth, and thus will help to bride the gap between theoretical and experimental research.  For future work we propose to scale our approach to larger graphs by means of SAT-based local improvement, a method that have been recently shown successful for the width parameters treewidth and branchwidth.

\end{abstract}
}

\section{Introduction} 

Graph decompositions have been a central topic in the area of
combinatorial algorithms, with applications in many areas of computer
science. A graph decomposition method is typically associated with a
\emph{width parameter} that indicates how well the given graph can be
decomposed. Tree decompositions, for instance, give rise to the width
parameter treewidth. In most cases, finding an \emph{optimal
  decomposition}, i.e., one of smallest width, is an NP-hard task, so
that for practical purposes one often relies on heuristics that
compute suboptimal decompositions. However, there are several reasons
why one is interested in optimal decompositions. If the purpose of the
decomposition is to facilitate the solution of a hard problem by means of
dynamic programming, then a suboptimal decomposition may impose an
exponential increase on time and space requirements for the dynamic
programming algorithm, and therefore may render the approach infeasible for the
instance under consideration. For instance. Kask et
al.~\cite{KaskGelfandOttenDechter11} noted about inference on
probabilistic networks of bounded treewidth: ``[\dots] since inference
is exponential in the tree-width, a small reduction in tree-width (say
even by $1$ or $2$) can amount to one or two orders of magnitude
reduction in inference time.''  Besides such algorithmic
applications, optimal decompositions are also useful for scientific
purposes, for instance to evaluate a heuristic method that provides an
upper bound on the decomposition width, or to support theoretical
investigations by facilitating the construction of gadgets for
hardness reductions.

An appealing approach to finding optimal decompositions are SAT-encodings, where one translates a given graph~$G$ and an integer $w$
into a propositional formula $F(G,w)$ whose satisfying assignments 
correspond to a decomposition of $G$ of width at most $w$. The satisfiability
of the formula can then be checked by a state-of-the art SAT-solver~\cite{MalikZhang09,MarquessilvaLynce14}. This approach was pioneered by Samer and
Veith~\cite{SamerVeith09} for treewidth. Their encoding was further
expanded on in subsequent works~\cite{BannachBerndtEhlers17,BergJarvisalo14} and today it still remains one of the most efficient methods for computing optimal tree decompositions. 
SAT-encodings have also
been developed for other graph parameters, including clique-width~\cite{HeuleSzeider15}, branchwidth~\cite{LodhaOrdyniakSzeider16}, as
well as pathwidth and special treewidth~\cite{LodhaOrdyniakSzeider17}.
This line of research revealed that the efficiency of the SAT-encoding based approach
crucially depends on the underlying characterisation of the considered
decompositional parameter. Whereas for treewidth the \emph{elimination-ordering}
based characterisations have been shown to be best suited for SAT-encodings, other
decomposition parameters require other characterisations. A
very efficient SAT-encoding for clique-width was based on the newly
developed \emph{partition-based} characterisation of clique-width~\cite{HeuleSzeider15}. Partition-based encodings have also been shown to be efficient for other width parameters~\cite{LodhaOrdyniakSzeider16,LodhaOrdyniakSzeider17}.
 
In this paper we develop SAT-encodings for the width parameters
\emph{treecut width} and \emph{treedepth}. These two parameter are
both less general than treewidth, i.e., any graph class where either
of these two parameters is bounded, is also of bounded treewidth, but
there exist graph classes of bounded treewidth where neither of these
two parameters are bounded. Neither of the two parameters (treecut width
and treedepth) is more general than the other, though. The parameters
are of interest as they offer certain algorithmic advantages over
treewidth; in particular, they support so-called \emph{fixed-parameter algorithms} for certain problems that are not \emph{fixed-parameter tractable} with respect to treewidth (see any of the handbooks on parameterized complexity~\cite{DowneyFellows13,CyganFKLMPPS14,FlumGrohe06}), as well as having a significantly lower \emph{parameter dependency} than treewidth for certain problems~\cite{GajHlin15,ElberfeldGT16}.

So far, both parameters have mainly been the subject of
theoretical investigations. By our encodings we provide the first
practical methods for computing the associated decompositions and
therefore provide a first step of bridging theoretical with
experimental research.

\subsection{Treecut Width}
The parameter treecut width was introduced by
Wollan~\cite{Wollan15}. Treecut width is an edge-separator based decompositional parameter whose relationship to the fundamental notion of \emph{graph immersions} is 
analogous to the relationship between treewidth and \emph{graph minors}~\cite{MarxWollan14}. Kim et al.~\cite{KimOPST18} gave a
linear time 2-approximation algorithm for treecut width, however, such
an error factor is prohibitive for practical use.  Ganian et
al.~\cite{GanianKimSzeider15,GanianKluteOrdyniak18} provided the first algorithmic results for
treecut width, and pointed out that several problems that are not
fixed-parameter tractable for the parameter treewidth are
fixed-parameter tractable for the parameter treecut width. 

Given that treecut width arguably has the most complicated and unintuitive characterisation
among all studied width parameters, our first step was
to find a way to simplify the definition of treecut width. Such
a simplification has recently been proposed by Kim et
al.~\cite{KimOPST18}, showing that the definition of treecut
decompositions becomes significantly more manageable on $3$-edge-connected
graphs and that computing decompositions for general graphs can
be reduced to the $3$-edge-connected case. Using this simpler
definition together with an explicit preprocessing procedure for
general graphs (presented in Section~\ref{ssec:pre-tcw}), we introduce a
SAT-encoding for $3$-edge-connected graphs based on a partition-based
characterisation of treecut width in
Section~\ref{sec:tcw}. As our experiments show, the encoding performs
extraordinary well; outperforming even our arguably much simpler encoding
for treedepth and the current best-performing SAT-encoding for treewidth~\cite{BannachBerndtEhlers17,BergJarvisalo14,SamerVeith09}.



\subsection{Treedepth} The parameter treedepth was introduced by Ne\v set\v ril and Ossona de
Mendez~\cite{NesetrilM06} in the context of their
graph sparsity project~\cite{NesetrilMendez12}. 
This parameter has been shown to have algorithmic applications for a number of problems where treewidth cannot be used. For instance,
Gutin et al.~\cite{GutinJonesWahlstrom15} showed that the Mixed
Chinese Postman problem is fixed-parameter tractable for treedepth,
but \cal{W}[1]-hard for treewidth and even pathwidth. Several further
algorithmic results on treedepth have been presented recently by Iwata
et al.~\cite{IwataOgasawaraOhsaka18}, Kouteck{\'{y}} et
al. \cite{KourteckyLevinOnn18}, Ganian and
Ordyniak~\cite{GanianOrdyniak18}, and Gajarsk{\'{y}} and
Hlin\v en{\'{y}}~\cite{GajarskyHlineny12}. Exact algorithms for
computing treedepth are known, e.g., the problem is known to be
fixed-parameter tractable~\cite{ReidlRVS14} and can be solved slightly
faster than $\mathcal{O}(2^n)$~\cite{FominGP15}, however, until now no
implementation of an exact algorithm for treedepth was available.

We introduce and implement two SAT-encodings for treedepth. The first
one explicitly guesses the tree-structure of a treedepth decomposition
and the second one is based on a novel partition-based
characterisation of treedepth. Since the partition-based encoding
greatly outperformed our first encoding, we mostly focus on the
partition-based encoding\sv{ and provide the first encoding only in the
the full version of the paper}. The experimental results for our partition-based encoding
are very promising, showing an extraordinarily good performance on
sparse classes of graphs such as paths, cycles, and complete binary
trees. We also introduce three novel preprocessing and symmetry
breaking procedures for treedepth.

\subsection{Related Work}
We have already mentioned the successful application of SAT-encodings
for graph decompositions above
\cite{BannachBerndtEhlers17,BergJarvisalo14,HeuleSzeider15,LodhaOrdyniakSzeider16,SamerVeith09}.
At this juncture we would like to briefly give some further context on
SAT-encodings. While every problem in NP admits a polynomial-time SAT-encoding,
it is well known that 
different encodings can behave quite differently in
practice~\cite{Prestwich09}.  There are some formal criteria which
indicate whether an encoding will behave well or not. However, only by
an experimental evaluation one can see what really works well and what
does not~\cite{Bjork09}.  For instance, while encoding size is certainly a
factor to take into consideration, larger encodings can work better if
they allow a fast propagation of conflicts, so that the power of
state-of-the art SAT-solvers which are based on the conflict-driven
clause learning paradigm (CDCL)
\cite{MalikZhang09,MarquessilvaLynce14} can be harvested. SAT-encodings are not only useful for the solution of hard combinatorial
problems in industry, such as the verification of hardware and
software \cite{Biere09}, but are increasingly often used in the context of
Combinatorics, for instance in the context of Ramsey
Theory \cite{Zhang09}. A very recent highlight is the celebrated
solution to the Pythagorean Triples Problem
\cite{HeuleKullmannMarek16,HeuleKullmann17}. 

Lastly, we would like to
mention that developing partition-based encodings 
comes with challenges that are specific to each width parameter,
as it
almost always requires the development of a novel characterization that is compatible
with such an encoding. 
Indeed, the existence of such an encoding for the probably most prominent width parameter, treewidth, remains open.
%


\section{Preliminaries}


We use $[i]$ to denote the set $\{0,1,\dots,i\}$. 
A \emph{weak partition} of
a set $S$ is a set $P$ of nonempty subsets of $S$ such that any two
sets in $P$ are disjoint; if additionally $S$ is the union of
all sets in $P$ we call $P$ a \emph{partition}. 
The
elements of~$P$ are called \emph{equivalence classes}.  Let $P,P'$ be
partitions of $S$. Then $P'$ is a \emph{refinement} of $P$ if for any
two elements $x,y\in S$ that are in the same equivalence class of $P'$
are also in the same equivalence class of $P$ (this entails the case
$P=P'$). 

\subsection{Formulas and Satisfiability}\label{ssec:pre-formula}
We consider propositional formulas in Conjunctive Normal Form
(\emph{CNF formulas}, for short), which are conjunctions of clauses,
where a clause is a disjunction of literals, and a literal is a
propositional variable or a negated propositional variable.  
A CNF
formula is \emph{satisfiable} if its variables can be assigned true or
false, such that each clause
contains either a variable set to true or
a negated variable set to false.  The satisfiability problem (SAT)
asks whether a given formula is satisfiable.

We will now introduce a few general assumptions and
notation that is shared among our encodings. Namely, for our encodings
we will assume that we are given an undirected graph $G=(V,E)$
and an integer $\omega$, which represents the width that
we are going to test. Moreover, we will assume that the vertices
of $G$ are numbered from $1$ to $n$ and similarly the edges are
numbered from $1$ to $m$.

For the counting part of our encodings
we will employ
the \emph{sequential counter} approach~\cite{SamerVeith09} since this
approach has turned out to provide the best results in our setting.
To illustrate the idea behind the sequential counter consider the case
that one is given a set $S$ of (propositional) variables and one
needs to restrict the number of variables in $S$ that are set to true
to be at most some integer $k$. For convenience, we refer to the
elements in $S$ using the numbers from $1$ to $|S|$. In this case one introduces
a counting variable $\ctr(s,j)$ for every $s \in S$ and $j$ with $1
\leq j \leq k$, which is true whenever there are at least $j$
variables in $\SB s' \SM s' \leq s\text{ and } s,s'\in S
\SE$ that are set to true. Then this can be ensured using the 
following clauses.
A clause $\neg s \lor \ctr(s,1)$
for every $s \in S$, 
a clause $\neg \ctr(s-1,j) \lor \ctr(s,j)$
for every $s\in S$ and $j$ with $s> 1$ and $1 \leq j
\leq k$, a clause $\neg s \lor \neg \ctr(s-1,j-1) \lor
\ctr(s,j)$
for
every $s\in S$ and $j$ with $s>1$ and $1 < j \leq k$, and
a clause $\neg s \lor \neg \ctr(s-1,k)$
for every $s \in S$ with $s > 1$.
This adds at most $\bigO(|S|k)$ variables and clauses to the
original formula.

\subsection{Graphs}

We use standard terminology for graph theory, see for
instance~\cite{Diestel10}. All graphs in this paper are undirected and
may contain multiedges.  Given a graph $G$, we let $V(G)$ denote its
vertex set and $E(G)$ its (multi-) set of edges.  The (open)
neighbourhood of a vertex $x \in V(G)$ is the set $\{y\in V(G):xy\in
E(G)\}$ and is denoted by $N_G(x)$. For a vertex subset $X$, the
neighbourhood of $X$ is defined as $\bigcup_{x\in X} N_G(x) \setminus
X$ and denoted by $N_G(X)$; we drop the subscript if the graph is
clear from the context.
For a vertex set
$A$ (or edge set $B$), we use $G-A$ ($G-B$) to denote the graph
obtained from $G$ by deleting all vertices in $A$ (edges in $B$), and
we use $G[A]$ to denote the \emph{subgraph induced on} $A$, i.e., $G-
(V(G)\setminus A)$.
Let $T$ be a rooted tree and $t \in V(T)$.
We 
write $T_t$ to denote the subtree of $T$ rooted in $t$, i.e., the
component of $T \setminus \{\{t,p\}\}$ containing $t$, where $p$ is
the parent of $t$ in $T$. We denote by $\mheight_T(t)$, the
\emph{height} of $t$ in $T$, i.e., the length of the path between the
root of $T$ and $t$ in $T$ plus one, and we denote by $\mheight(T)$
the \emph{height} of
$T$, i.e., the maximum of $\mheight_T(t')$ over all $t' \in V(T)$.
Let $G$ be a graph. We say that two vertices $u$ and $v$ of $G$ are
\emph{$3$-edge-connected} in $G$ if there are at least $3$ pairwise
edge disjoint paths between $u$ and $v$ in $G$. We say a subset $C$ of
$V(G)$ is a \emph{$3$-edge-connected component} of $G$ if every pair
of distinct vertices in $C$ is $3$-edge-connected and $C$ is maximal
w.r.t.\ this property.
For a graph $G$ and a subset $V' \subseteq V(G)$, we denote by $\delta_G(V')$ the
(multi-)set of edges of $G$ having one endpoint in $V'$ and one
endpoint in $V(G) \setminus V'$ and omit the subscript $G$ if it can
be inferred from the context. An \emph{apex vertex} is a vertex adjacent to all other vertices
in the graph.


\subsection{Treecut Width}
\label{ssec:pre-tcw}
The notion of treecut width and treecut decomposition was originally
introduced for general graphs~\cite{MarxWollan14,Wollan15}. Here we
use a simpler definition, which allows for an easier encoding, and
only applies for $3$-edge-connected graphs.
Using known results~\cite{KimOPST18}, we will then show that
the treecut width for general graphs can be defined in terms of the
treecut widths of its $3$-edge-connected components.
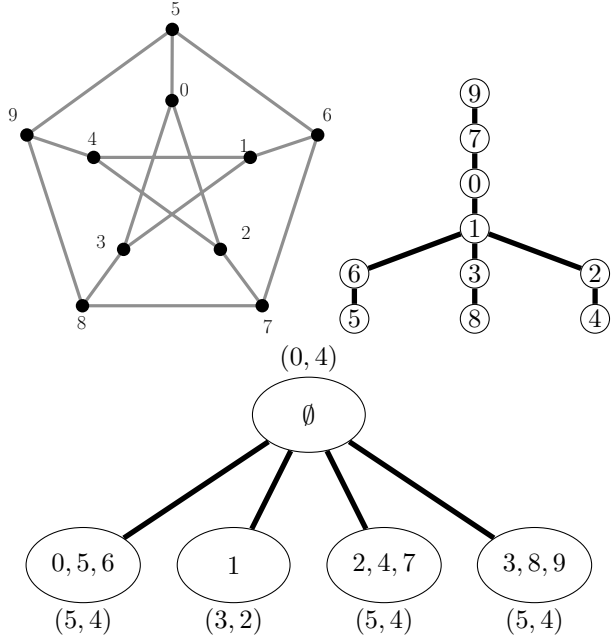
\begin{figure}[ht]
  \begin{center}
  \includegraphics[width=0.25\textwidth]{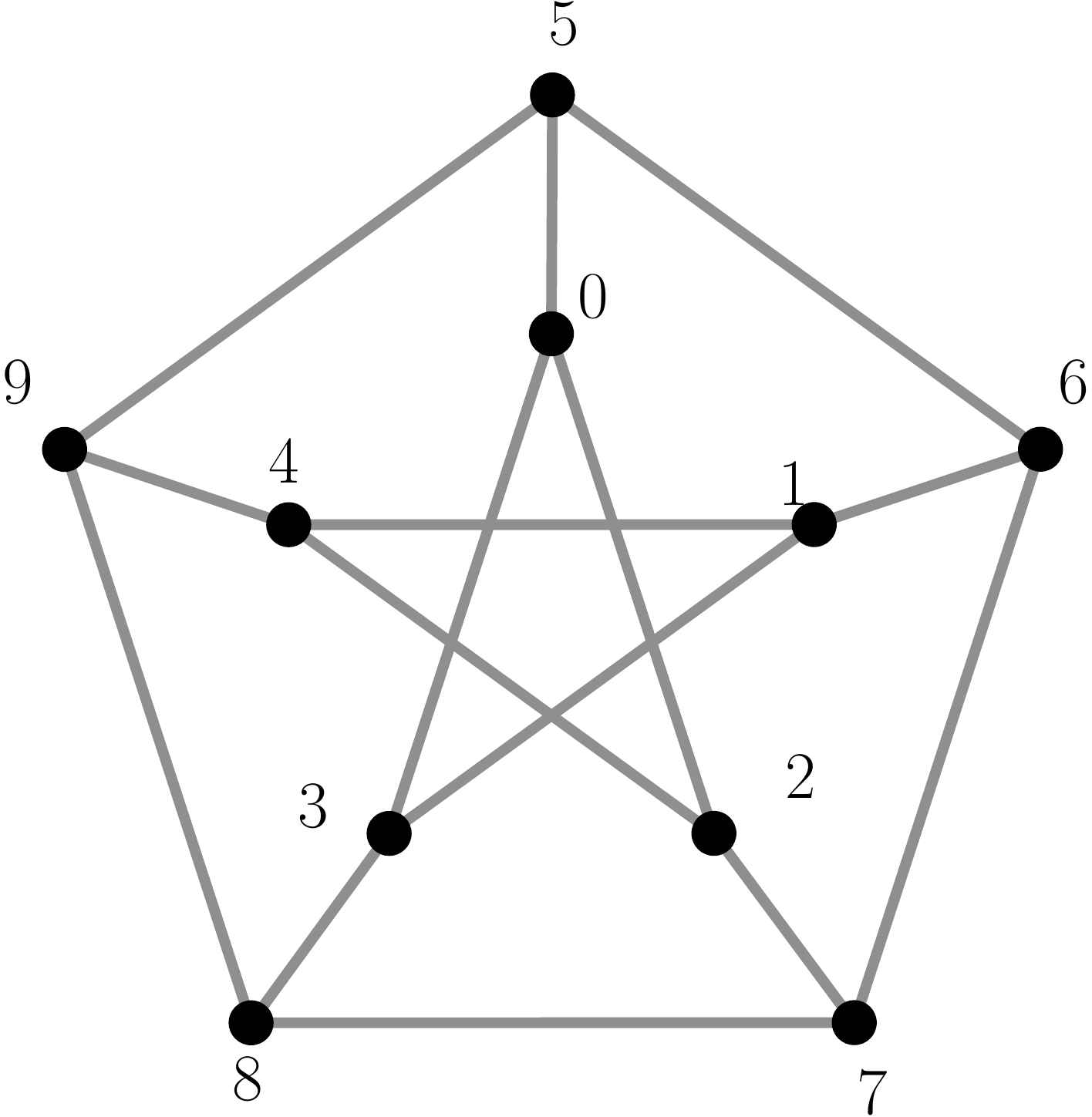}
  \begin{tikzpicture}[scale=1, node distance=0.6cm]
    \tikzstyle{tn} = [draw, circle, inner sep=1pt]
    \tikzstyle{ln} = []
    \tikzstyle{te} = [draw, line width=2pt]
    
    \draw
    node[tn] (9) {$9$}
    node[tn, below of=9] (7) {$7$}
    node[tn, below of=7] (0) {$0$}
    node[tn, below of=0] (1) {$1$}
    
    node[tn, below of=1] (3) {$3$}
    node[tn, below of=3] (8) {$8$}
    
    node[tn, node distance=1.6cm, right of=3] (2) {$2$}
    node[tn, below of=2] (4) {$4$}
    
    node[tn, node distance=1.6cm, left of=3] (6) {$6$}
    node[tn, below of=6] (5) {$5$}
    ;
    
    \draw
    (9) edge[te] (7)
    (7) edge[te] (0)
    (0) edge[te] (1)
    
    (1) edge[te] (3)
    (3) edge[te] (8)
    
    (1) edge[te] (2)
    (2) edge[te] (4)
    
    (1) edge[te] (6)
    (6) edge[te] (5)
    ;
  \end{tikzpicture}
  \quad\quad\quad\quad
  \begin{tikzpicture}[scale=1, node distance=2cm]
      \tikzstyle{tn} = [draw, ellipse, minimum height=1cm, minimum width=1.5cm]
      \tikzstyle{ln} = []
      \tikzstyle{te} = [draw, line width=2pt]
      
      \draw
      node[tn] (r) {$\emptyset$}
      node[node distance=2cm, below of=r] (br) {}
      node[tn, node distance=1cm, left of=br] (l2) {$1$}
      node[tn, left of=l2] (l1) {$0,5,6$}
      node[tn, right of=l2] (l3) {$2,4,7$}
      node[tn, right of=l3] (l4) {$3,8,9$}
      ;
      
      \draw[node distance=0.75cm]
      node[ln, below of=l1] (l1l) {$(5,4)$}
      node[ln, below of=l2] (l2l) {$(3,2)$}
      node[ln, below of=l3] (l3l) {$(5,4)$}
      node[ln, below of=l4] (l4l) {$(5,4)$}
      node[ln, above of=r] (rl) {$(0,4)$}
      ;
      
      \draw
      (r) edge[te] (l1)
      (r) edge[te] (l2)
      (r) edge[te] (l3)
      (r) edge[te] (l4)
      ;
    \end{tikzpicture}
  \end{center}
  \caption{A width-$6$ treedepth decomposition (top right) and
    a width-$5$ treecut decomposition (bottom) of the Petersen
    graph (top left). The treecut decomposition lists the adhesion (left value) and torso-width
    (right value) of each node.}
  \label{fig:tcw-ex}
\end{figure}

Let $G$ be a $3$-edge connected undirected graph (possibly with multi-edges and loops).
A \emph{treecut decomposition} of $G$ is a pair $(T,\chi)$, where $T$
is a rooted tree and $\chi : V(T) \rightarrow 2^{V(G)}$ such that $\SB
\chi(t) \SM t\in V(T)\SE$ forms a \emph{near partition} of $V(G)$,
i.e., a partition of $V(G)$ allowed to contain the empty set.
For a subgraph $T'$ of $T$, we denote by $\chi(T')$ the set
$\bigcup_{t \in V(T')}\chi(t)$.
Let $t \in V(T)$. We denote by $V_t$ the set $\chi(T_t)$. The
\emph{adhesion} of $t$, denoted by $\adhesion(t)$, is the (multi-)set
$\delta_G(V_t)$. Moreover, the \emph{torsowidth} of $t$, denoted by
$\torsow(t)$, is equal to $|\chi(t)|$ plus the number of neighbours of $t$ in
$T$. The \emph{width} of $(T,\chi)$ is the maximum width of
any of its nodes $t \in V(T)$, which in turn is equal to
$\max\{|\adhesion(t)|,\torsow(t)\}$. The \emph{height} of
$(T,\chi)$ is simply the height of $T$. Finally, the \emph{treecut width}
of $G$, denoted by $\tcw(G)$, is the minimum width of any of its treecut decompositions.
Figure~\ref{fig:tcw-ex} illustrates a treecut decomposition for the Peterson
graph.

The following lemma shows that if a graph is not $3$-edge-connected,
then it can be modified and split into parts in such a way that its
treecut width can be computed from the treecut width of the
(modified) parts. Since a recursive application of this lemma
eventually results in $3$-edge-connected graphs, the lemma allows us
to apply our encoding for $3$-edge-connected graphs to arbitrary graphs.
\begin{lemma}
    \label{lem:tcw-3ec}
  Let $G$ be a multigraph, $C$ be a minimal cut of size at most
  two resulting in the partition $(A,B)$ of $V(G)$, and let $A_C$ and
  $B_C$ be the endpoints of the edges in $C$ in $A$ and $B$,
  respectively. If $C$ contains two edges and $|A_C|=|B_C|=1$,
  then $\tcw(G)=\max\{2,\tcw(G[A]),\tcw(G[B])\}$. Otherwise,
  $\tcw(G)=\max\{\tcw(G_A),\tcw(G_B)\}$, where $G_A$ ($G_B$) is obtained
  from $G[A]$ ($G[B]$) after adding an edge between the vertices in $A_C$ ($B_C$);
  note that an edge is only added if $|A_C|=2$ or $|B_C|=2$, respectively.
\end{lemma}
\begin{proof}
  The proof is closely based on the ideas in~\cite[Section
  3]{KimOPST18}. Namely, in the case that $C$ does not contain two
  edges sharing the same endpoints, the proof follows immediately
  from~\cite[Lemma 3 and 4]{KimOPST18}. Moreover, if $C$ contains two
  edges sharing the same endpoints, say $a \in A$ and $b \in B$, it
  follows from~\cite[Lemma 3]{KimOPST18} that $\tcw(G)=\max
  \{\tcw(G[A\cup \{b\}]),\tcw(G[B \cup \{a\}])$. Moreover, using the
  definition of treecut width for arbitrary graphs and recalling that $b$ has precisely $2$ neighbours in $A$ (and similarly $a$ has precisely $2$ neighbours in $B$), it is then easy to
  see that $\tcw(G[A\cup\{b\}])=\max\{2,\tcw(G[A])\}$ and similarly
  $\tcw(G[B\cup\{a\}])=\max\{2,\tcw(G[B])\}$, from which the lemma follows.   
  More precisely, this follows immediately by observing that (1)
  $\tcw(G[\{a,b\}])=2$ and (2) a treecut decomposition $(T,\chi)$ of $G[A]$
  ($G[B]$) of width $w$ can be turned into a treecut decomposition of
  $G[A\cup\{b\}]$ ($G[B\cup\{a\}]$) of width $\max\{2,w\}$ by adding
  a leaf $l$ containing $b$ ($a$) as a leaf to an arbitrary node of
  $T$. Note that $l$ has torsowidth at most $2$ and the torsowidth of the neighbor of $l$ in $T$
  is not increased; to see this one needs to use the definition
  of treecut width on general graphs and the fact that $l$ is a
  \emph{thin node}~\cite{GanianKimSzeider15}. \qed
\end{proof}
We also give the known relations between treecut width, treewidth, and maximum degree.
\begin{lemma}[\cite{GanianKimSzeider15,MarxWollan14,Wollan15}]\label{lem:tcw-rel}
  For every graph $G$, $\tw(G)\leq 2\tcw(G)^2+3\tcw(G)$ and
  $\tcw(G)\leq 4\Delta(G) \cdot \tw(G)$, where $\Delta(G)$ and
  $\tw(G)$ denote the maximum degree and treewidth of $G$, respectively.
\end{lemma}
We close this section by showing explicit values of treecut width for
complete graphs ($K_n$) and complete bipartite graphs ($K_{n,n}$), which
we later employ to verify the correctness of our encoding. \lv{For the
proofs we will assume w.l.o.g. that a treecut decomposition $(T,\chi)$
does not contain unnecessary nodes, i.e., $\chi(l)\neq \emptyset$ if $l$ has at most one child in $T$.}
\sv{\begin{lemma}[$\star$]
  \label{lem:tcw-comp-sv}
  For every $n\geq 3$, it holds that $\tcw(K_{n+1})=n+1$ and $\tcw(K_{n,n})=2n-2$.
\end{lemma}}
\lv{\begin{lemma}
  \label{lem:tcw-comp}
  For every $n\geq 4$, $\tcw(K_n)=n$.
\end{lemma}
\begin{proof}
First, note that $\tcw(K_n)\leq n$ since there is a trivial treecut decomposition of width $n$. So, assume for a contradiction that there exists a treecut decomposition $(T,\chi)$ of $K_n$ of width smaller than $n$; without loss of generality, we assume that $\chi(t)\neq \emptyset$ for each leaf $t$ of $T$, and similarly $\chi(t')\neq \emptyset$ for each parent $t'$ of a single leaf in $T$. Since $n\geq 4$, in order not to exceed the bound on the width due to adhesion, for each edge $e$ of $T$ it must hold that one tree in $T-e$ must contain a single node (i.e., a leaf) $u$ and $|\chi(u)|=1$. From this it follows that $T$ must be a star (with center $r$). However, since the torsowidth of $r$ is equal to $|\chi(r)|$ plus the number of leaves (each representing a single vertex of $K_n$), we see that $\torsow(r)=n$, a contradiction.  
\qed
\end{proof}
\begin{lemma}
  \label{lem:tcw-comp-bip}
  For every $n\geq 3$, $\tcw(K_{n,n})=2n-2$.
\end{lemma}
\begin{proof}
Let $V(K_{n,n})=\{a_1,\dots,a_n,b_1,\dots,b_n\}$ and $E(K_{n,n}=\SB a_ib_j \SM i,j\in [n] \SE$. We obtain a treecut decomposition of $K_{n,n}$ of width $2n-2$ as follows: $T'$ is a star with center $r$ and leaves $t_1,\dots,t_n$, where $\chi'(r)=\emptyset$ and $\chi'(t_i)=\{a_i,b_i\}$ for all $i\in [n]$.

Now, assume for a contradiction that there exists a treecut decomposition $(T,\chi)$ of width smaller than $2n-2$. As before, we assume that $\chi(t)\neq \emptyset$ for each leaf $t$ of $T$, and similarly $\chi(t')\neq \emptyset$ for each parent $t'$ of a single leaf in $T$. Since $n\geq 3$, in order not to exceed the bound on the width due to adhesion, for each edge $e$ of $T$ it must hold that one tree in $T-e$ must contain a single node (i.e., a leaf) $u$ and $|\chi(u)|=1$. As before, from this it follows that $T$ must be a star (with center $r$). However, since the torsowidth of $r$ is equal to $|\chi(r)|$ plus the number of leaves (each representing a single vertex of $K_n$), we see that in fact $\torsow(r)=2n$, a contradiction.
%
\qed
\end{proof}
  }
\subsection{Treedepth}

\newcommand{\undC}{C}

The second decompositional parameter for which we will introduce a
SAT-encoding is treedepth~\cite{NesetrilMendez12}.
Treedepth is closely related to treewidth, and the structure of graphs
of bounded treedepth is well understood~\cite{NesetrilMendez12}.
A useful way of thinking about graphs of bounded treedepth is that they are
(sparse) graphs with no long paths.

      
      


      
      



The \emph{treedepth} of an undirected graph $G$, denoted by $\td(G)$,
is the smallest natural number $k$ such
that there is an undirected rooted forest $F$ with vertex set $V(G)$ of height
at most $k$ for which $G$ is a subgraph of $\undC(F)$, where
$\undC(F)$ is called the \emph{closure} of $F$ and is the undirected graph with vertex set $V(F)$ having an
edge between $u$ and $v$ if and only if $u$ is an ancestor of $v$ in
$F$. A forest $F$ for which $G$ is a subgraph of $\undC(F)$ is also
called a \emph{treedepth decomposition}, whose \emph{depth} is equal
to the height of the forest.
Informally a graph has treedepth at most $k$ if it can be
embedded in the closure of a forest of height $k$. Note that if $G$ is
connected, then it can be embedded in the closure of a tree instead of
a forest. A treedepth decomposition of the Peterson graph is
illustrated in Figure~\ref{fig:tcw-ex}.

\noindent We conclude with some useful facts about treedepth. 

\begin{lemma}[\cite{NesetrilMendez12}]
  \label{pro:tdfacts}
  For every graph $G$, $\tw(G)\leq \td(G)$ and $\pw(G) \leq \td(G)$, where $\pw(G)$
  is the pathwidth of $G$.
\end{lemma} 

\section{Treecut Width}\label{sec:tcw}

In~this section we will introduce our encoding for treecut width. The
encoding is based on a different characterisation of treecut
width, one that is well-suited for SAT-encodings.

\subsection{Partition-Based Formulation}\label{sec:part-ch}

Here we present a partition-based characterisation of treecut width,
in terms of what we call derivations, which is well-suited for an encoding into SAT. 
Let $G$ be a graph. A \emph{derivation} $\der$ of $G$ of \emph{length}
$l$ is a sequence $(P_1,\dotsc,P_l)$ of weak partitions of $V(G)$ such
that:
\begin{itemize}
\item[D1] $P_1=\emptyset$ and $P_l=\{\{V(G)\}\}$ and
\item[D2] for every $i\in \{1, \dots, l\}$, $P_i$ is a refinement of
  $P_{i+1}$.
\end{itemize}

\newcommand{\bagder}[3]{\chi_{#1}^{#2}(#3)}
\newcommand{\chder}[3]{\mathsf{c}_{#1}^{#2}(#3)}
\newcommand{\torwder}[3]{\torsow_{#1}^{#2}(#3)}

We will refer to $P_i$ as
the $i$-th \emph{level} of the derivation $\der$ and we will refer to
elements in $\bigcup_{1\leq i\leq l}P_i$ as \emph{sets} of the
derivation. Let $p \in P_i$ for some level $i$ with $1 \leq i \leq
l$. We say that a set $c \in P_{i-1}$ is a \emph{child} of $p$ at level $i$
if $c\subseteq p$ and denote by $\chder{\der}{i}{p}$ the set of all
children of $p$ at level $i$. Moreover, we denote by
$\bagder{\der}{i}{p}$ the set $p \setminus (\bigcup_{c \in
  \chder{\der}{i}{p}}c)$. Then the \emph{width} of $p$ at level $i$ is
equal to the maximum of $|\delta_G(p)|$ and
$\torwder{\der}{p}{i}$, where $\torwder{\der}{p}{i}$ is equal to
$|\bagder{\der}{i}{p}|+|\chder{\der}{i}{p}|+1$ if $i \neq l$ and 
equal to $|\bagder{\der}{i}{p}|+|\chder{\der}{i}{p}|$ otherwise.
We will show that any treecut decomposition can be
transformed into a derivation of the same width, and vice versa.
The following example illustrates the close connection
between treecut decompositions and derivations.

\medskip 
\noindent \emph{Example:}
  The treecut decomposition given in
  Fig.~\ref{fig:tcw-ex} of the Petersen graph can be translated into
  the derivation $\mathcal{P}=(P_1,\dotsc,P_3)$ defined by:
  \smallskip

\noindent    $P_1=\emptyset$, $P_2=\Big\{\big\{0,5,6\big\},\big\{1\big\},\big\{2,4,7\big\},\big\{3,8,9\big\}\Big\}$, \\
$P_3=\Big\{\big\{0,1,2,3,4,5,6,7,8,9\big\}\Big\}$.


  \smallskip
  \noindent As can be verified easily, the width of $\mathcal{P}$ is equal to $5$.
\medskip

We show that derivations provide an alternative
characterisation of treecut decompositions.
\lv{\begin{theorem}}
  \sv{\begin{theorem}[$\star$]}
  \label{the:equi-tc-der}
  Let $G$ be a graph and $\omega$ and $d$ two integers. 
  $G$ has a treecut decomposition of
  width at most $\omega$ and height at most $d$ if and only if $G$ has a derivation
  of width at most $\omega$ and length at most $d+1$.
\end{theorem}
\lv{\begin{proof}
  Let $(T,\chi)$ be a treecut decomposition of $G$ of width at most
  $\omega$ and height at most $d$; without loss of generality, we assume that $\chi(t)\neq \emptyset$ for each leaf $t$ of $T$. 
  It is immediate from the
  definitions that
  $\der=(P_1,\dotsc,P_{\mheight(T)+1})$ such that $P_1=\emptyset$ and $P_i=\SB V_t \SM t \in
  V(T) \land \mheight_T(t)=\mheight(T)-i+2\SE$ for every $i$ with $2\leq i \leq \mheight(T)+1$
  is a derivation of $G$
  with width at most $\omega$ and length at most $d+1$.

  Towards showing the converse, let $\der=(P_1,\dotsc,P_l)$ be a
  derivation of $G$ with width at most $\omega$.
  It is immediate from the definitions that $(T,\chi)$ such that:
  \begin{itemize}
  \item $T$ is the tree with a vertex for every pair $(p,i)$ such that $p
    \in P_i$ having an edge between $(c,i-1)$ and $(p,i)$ if $c$ is a
    child of $p$ at level $i$, and
  \item $\chi((p,i))=\bagder{\der}{i}{p}$ for every level $i$ and $p
    \in P_i$
  \end{itemize}
  is a treecut decomposition of $G$ with width at most $\omega$ and height
  at most $d-1$.
  \qed
\end{proof}
}

\subsection{Encoding}

Let $G$ be a graph with $m$ edges and $n$ vertices, and let $\omega$
and $d$ be positive integers. We will assume that the vertices of $G$
are represented by the numbers from $1$ to $n$ and the edges of $G$ by
the numbers from $1$ to $m$.  The aim of this section is to construct
a formula $F(G,\omega,d)$ that is satisfiable if and only if $G$ has a
derivation of width at most $\omega$ and length at most $d$. 
Because of
Theorem~\ref{the:equi-tc-der} (after setting $d$ to $n$) it holds that $F(G,\omega,d)$ is satisfiable if and only if
$G$ has treecut width at most $\omega$.  To achieve this aim we first
construct a formula $F(G,d)$ such that every satisfiable assignment encodes a derivation
of length at most $d$ and then we extend this formula
by adding constrains that restrict the width of the derivation to $\omega$.

\subsubsection{Encoding of a Derivation}\label{sec:tcw-encoding}

The formula $F(G,d)$ uses the following variables.  A \emph{set
  variable} $\comp(u,v,i)$, for every $u,v \in V(G)$ with $u\leq v$ and
every $i$ with $1\leq i \leq d$.  Informally, $\comp(u,v,i)$ is true
whenever $u$ and $v$ are contained in the same set at level $i$ of the
derivation. Note that $\comp(u,u,i)$ is true whenever $u$ is contained
in some set at level $i$. Furthermore, the formula contains a \emph{leader variable} $\led(u,i)$, for every
$u \in V(G)$ and every $i$ with $1\leq i \leq d$.  Informally, the
leader variables will be used to uniquely identify the sets at each
level of a derivation (using the smallest vertex contained in the set
as the unique identifier), i.e., $\led(u,i)$ is true whenever $u$ is the
smallest vertex in a set at level $i$ of the derivation.

We now describe the clauses of the formula.
The following clauses ensure (D1) and (D2).
\begin{tabbing}
  \quad \quad \=  \kill
  \> $\neg \comp(u,v,1) \wedge \comp(u,v,d)$ \` for $u,v \in V(G)$, $u \leq v$\\
 \> $\neg \comp(u,v,i)\vee \comp(u,v,i+1)$ \` ~ \\
  \> {\color{white}*******} \`for $u,v \in V(G)$, $u\leq v$, $1 \leq i < d$  
\end{tabbing}

\noindent The following clauses ensure that if a vertex $v$ is in some
set with at least one other vertex at
level $i$, then $\comp(v,v,i)$ is true.
\begin{tabbing}
  \quad \quad \= \kill
  \> $(\neg \comp(u,v,i) \vee \comp(u,u,i))\wedge (\neg \comp(u,v,i)
  \vee \comp(v,v,i))$ \` ~\\
  \> {\color{white}*******} \` for $u,v \in V(G)$, $u< v$, $2 \leq i \leq d$\\
\end{tabbing}
\vspace{-6mm}
\noindent The following clauses ensure that the relation of being in the
same set is transitive.
\begin{tabbing}
  \quad \quad \= \kill
  \> $(\neg \comp(u,v,i)\vee\neg \comp(u,w,i)\vee \comp(v,w,i))$ \\
  \> $\wedge (\neg \comp(u,v,i)\vee \neg \comp(v,w,i)\vee \comp(u,w,i))$\\
  \> $ \wedge (\neg \comp(u,w,i)\vee\neg  \comp(v,w,i)\vee \comp(u,v,i))$ \\
    \> {\color{white}*******} \` for $u,v,w \in V(G)$, $u<v<w$, $1 \leq i \leq d$\\
\end{tabbing}

\noindent The following clauses ensure that $\led(u,i)$ is true if and only if $u$ is
the smallest vertex contained in some set at level $i$ of a derivation.
\begin{tabbing}
  \quad \quad \= \kill
  \> (A) \quad $(\led(u,i) \vee \neg \comp(u,u,i) \vee \bigvee\limits_{v\in
    V(G),v<u}\comp(v,u,i))\wedge$\\
  \> (B) \quad $(\neg \led(u,i) \vee \comp(u,u,i)) \wedge \bigwedge\limits_{v\in
    V(G),v<u}(\neg {\led}(u,i)\vee$ \\ 
    \> {\color{white} (*)} \quad $\neg \comp(v,u,i))$ 
\end{tabbing}
\vspace{-6mm}
\phantom{x}\hfill for $u \in V(G)$, $1 \leq i \leq d$

\noindent Part $A$ ensures that if $u$ is contained in some set at
level $i$ and no vertex smaller than $u$ is contained in a set with
$u$ at level $i$, then $u$ is a leader. Part $B$ ensures that if $u$
is a leader at level $i$, then $u$ is contained in some set at level
$i$ and furthermore no vertex smaller than $u$ at level $i$
is in the same set as $u$.
The
formula $F(H,d)$ contains at most $\bigO(n^2d)$ variables and
$\bigO(n^3d)$ clauses.


\subsubsection{Encoding of a Derivation of Bounded Width}

\newcommand{\adh}{\textsf{ad}}
\newcommand{\tor}{\textsf{tor}}

Next,~we describe how $F(G,d)$ can be extended to restrict the width of
the derivation. Towards this aim we first need new
variables allowing us to define adhesion and torsowidth. Namely, for
every $u \in V(G)$, $e \in E(G)$ such that at least one of the
endpoints of $e$ is larger or equal to $u$, and $i \in \{2,\dotsc,d-1\}$, we use
the variable $\adh(u,e,i)$, which will be true if $u$ is a leader of
some set $V'$ at level $i$ and $e \in \delta_G(V')$. This is
ensured by the following clauses. 
\begin{tabbing}
  \quad \quad \= \kill
  \> $\neg \led(u,i) \vee \neg \comp(u,v,i) \vee \comp(u,w,i) \vee 
\adh(u,e,i)$
\end{tabbing}
\vspace{-6mm}
\phantom{x}\hfill for $u,v,w \in V(G)$, $e=\{v,w\}\in E(G)$, $u \leq v$, \\
\phantom{x}\hfill $u\leq w$, and $1 < i < d$.

\begin{tabbing}
  \quad \quad \= \kill
  \> $\neg \led(u,i) \vee \neg \comp(u,v,i) \vee 
\adh(u,e,i)$
\end{tabbing}
\vspace{-6mm}
\phantom{x}\hfill for $u,v,w \in V(G)$, $e=\{v,w\}\in E(G)$, $u \leq
v$, \\
\phantom{x}\hfill $w < u$, and $1 < i < d$.

Note that we
do not require the reverse direction here since the sole purpose of
the variables $\adh(u,e,i)$ is to ensure that the adhesion never
exceeds the width. 

Towards defining torsowidth, we introduce the variable $\tor(u,v,i)$
for every $u,v \in V(G)$, $u \leq v$, and $1 \leq i \leq d$, which
will be true, whenever $u$ is a leader of a set $V'$ at level $i$, $v$
is in $V'$, and either $v$ is a leader at level $i-1$, or
$v$ is not in a set at level $i-1$.
This is
ensured by the following clauses. 
\begin{tabbing}
  \quad \quad \= \kill
  \> $\neg \led(u,i) \vee \neg \comp(u,v,i) \vee \neg \led(v,i-1) \vee 
\tor(u,e,i)$ \` \\
\> {\color{white}**} \` for $u,v \in V(G)$, $u \leq v$, and $2 < i \leq d$\\
  \> $\neg \led(u,i) \vee \neg \comp(u,v,i) \vee \neg \comp(v,v,i-1) \vee 
\tor(u,e,i)$ \`\\
\> {\color{white}**} \` for $u,v \in V(G)$, $u \leq v$, and $1 < i \leq d$
\end{tabbing}

\lv{Note that after adding the above variables and clauses defining adhesion and
torsowidth, our formula has at most at most $\bigO(n^2d+nmd)$ variables and
$\bigO(n^3d)$ clauses.}

Finally, we use the sequential counter introduced in Section~\ref{ssec:pre-formula} to ensure that
both the adhesion as well as the torsowidth do not exceed the given
width. Namely, for every $u \in V(G)$ and $i$ with $1 < i < d$, we
ensure that the number of variables in $\SB \adh(u,e,i) \SM e \in E(G)
\SE$ that are set to true does not exceed $\omega$. Similarly for every $u \in V(G)$
and $i$ with $1 < i < d$, we ensure that the number of variables in
$\SB \tor(u,v,i)\SM v \in V(G) \land u\leq v\SE$ that are set to true
does not exceed $\omega-1$ and that the number of variables in $\SB
\tor(u,v,d)\SM v \in V(G) \land u\leq v\SE$ does not exceed $\omega$.

\sloppypar
This completes the construction of $F(G,\omega,d)$.
By construction, $F(G,\omega,d)$ is satisfiable
if and only $G$ has a derivation of width at most $\omega$ and length at most
$d$. Due to Theorem~\ref{the:equi-tc-der}, we obtain:
\begin{theorem}\label{the:formula-tc}
  The formula $F(G,\omega,d)$ is satisfiable if and only if $G$ has a treecut
  decomposition of width at most $\omega$ and depth at most $d$. Moreover,
  such a treecut decomposition can be constructed from a satisfying
  assignment of $F(G,\omega,d)$ in linear time w.r.t.\ the number of
  variables of $F(G,\omega,d)$. 
\end{theorem}

\section{Treedepth}

In~this section we introduce a SAT-encoding for treedepth, which is
also based on partitions. We also developed an encoding for treedepth
that is based on guessing the tree of the treedepth decomposition,
however, the encoding has, to our surprise, performed much worse than
the partition-based encoding. Namely, the encoding, which we introduce for
completeness in \sv{the full version of the paper}\lv{Section~\ref{asec:td-enc-2}}, only terminated on 17 out of the 39
famous graphs.

\subsection{Partition-Based Formulation}

Let $G$ be a graph. We will base our definition of derivations for
treedepth on the derivations defined for treecut width in
Section~\ref{sec:part-ch}. A \emph{derivation} $\der$ of $G$ of length
$l$ is a sequence $(P_1,\dotsc,P_l)$ of weak partitions of $V(G)$
satisfying (D1) and (D2) and additionally the following properties:
\begin{itemize}
\item[(D3)] for every $p \in P_i$, $|\bagder{\der}{i}{p}|\leq 1$, and
\item[(D4)] for every edge $\{u,v\} \in E(G)$, there is a $p\in P_i$
  such that $\{u,v\} \subseteq p$ and $\bagder{\der}{i}{p}\cap
  \{u,v\}\neq \emptyset$.
\end{itemize}
It will be useful to recall the notions defined for derivations in Section~\ref{sec:part-ch}.
We will show that any treedepth decomposition of depth $\omega$ can be
transformed into a derivation of length $\omega+1$, and vice versa.
The following example illustrates the connection
between treedepth and such derivations.

\medskip
\noindent \emph{Example:}
  The treedepth decomposition given in
  Fig.~\ref{fig:tcw-ex} of the Petersen graph can be translated into 
  the derivation $\mathcal{P}=(P_1,\dotsc,P_7)$ defined by:

  \smallskip
$P_1=\emptyset$, $P_2=\Big\{\big\{4\big\},\big\{8\big\},\big\{5\big\}\Big\}$,

$P_3=\Big\{\big\{2,4\big\} \big\{3,8\big\},\big\{5,6\big\}\Big\}$,

$P_4=\Big\{\big\{1,2,3,4,5,6,8\big\}\Big\}$, 
  
  $P_5=\Big\{\big\{0,1,2,3,4,5,6,8\big\}\Big\}$
  
  $P_6=\Big\{\big\{0,1,2,3,4,5,6,7,8\big\}\Big\}$, 
  
  $P_7=\Big\{\big\{0,1,2,3,4,5,6,7,8,9\big\}\Big\}$.
\medskip

\noindent The next theorem shows that such derivations provide an alternative
characterisation of treedepth.
\lv{\begin{theorem}}
\sv{\begin{theorem}[$\star$]}
  \label{the:equi-td-der}
  Let $G$ be a connected graph and $\omega$ an integer.
  $G$ has a treedepth decomposition of
  depth at most $\omega$ if and only if $G$ has a derivation
  of length at most $\omega+1$.
\end{theorem}
\lv{\begin{proof}
  Let $T$ be a treedepth decomposition of $G$. It is immediate from the
  definitions that
  $\der=(P_1,\dotsc,P_{\mheight(T)+1})$ such that $P_1=\emptyset$ and $P_i=\SB V(T_t) \SM t \in
  V(T) \land \mheight_T(t)=\mheight(T)-i+2\SE$ for every $i$ with $2\leq i \leq \mheight(T)+1$
  is a derivation of $G$
  whose length is equal to the height of $T$ plus~1.

  Towards showing the converse, let $\der=(P_1,\dotsc,P_l)$ be a
  derivation of $G$.
  Note first that w.l.o.g. we can assume that $\bagder{\der}{i}{p}\neq
  \emptyset$ for every $i$ with $ 1\leq i \leq l$ and $p\in P_i$; this
  is because if this is not the case we can replace $p$ in $P_i$ with
  all its children in $P_{i-1}$. 
  Note also that for every $v \in V(G)$, there is exactly one $p \in P_i$ such that
  $\{v\}=\bagder{\der}{i}{p}$, which we will call the set of $\der$ introducing
  $v$.
  Let $T$ be the tree with vertex set
  $V(G)$ having an edge between $u$ and $v$, whenever the set
  introducing $u$ is a child of the set introducing $v$ or vice versa 
  and let $\bagder{der}{l}{r}$ with $r \in P_l$ be the root of $T$. It
  is straightforward to verify that $T$ is a treedepth decomposition
  of $G$ with depth at most $l-1$. 
  \qed
\end{proof}
}
\subsection{Encoding of a Derivation}
Here we construct the formula $F(G,\omega)$ that is satisfiable if and only if $G$ has a
derivation of length at most $\omega$, which together with
Theorem~\ref{the:equi-td-der} implies that $G$ has treedepth
$\omega-1$. Since we are again using derivations, the formula
$F(G,\omega)$ is relatively similar to the formula $F(G,d)$ introduced
in Section~\ref{sec:tcw-encoding}. In particular, we again have a \emph{set
  variable} $\comp(u,v,i)$, for every $u,v \in V(G)$ with $u\leq v$ and
every $i$ with $1\leq i \leq \omega$, which has the same semantics as
before. It also contains all the clauses introduced in
Section~\ref{sec:tcw-encoding} apart from the clauses restricting the
leader variables. Additionally, we have the following clauses, which ensure that (D3) holds, i.e., that there is at most one vertex in $\bagder{\der}{i}{p}$.
\begin{tabbing}
\quad\quad\=\kill
\>$\neg\comp(u,v,i)\vee\comp(u,u,i-1)\vee\comp(v,v,i-1)$\`\\
\>{\color{white}***}\` for all $u,v\in V$, 
$u<v$, and $2\leq i\leq\omega$
\end{tabbing}
Finally, the following clauses ensure (D4).
\begin{tabbing}
 \quad\quad\=\kill
 \>$\neg\comp(u,u,i)\vee\neg\comp(v,v,i)\vee\comp(u,u,i-1)\vee\comp(u,v,i)$\\
 
\>$\neg\comp(u,u,i)\vee\neg\comp(v,v,i)\vee\comp(v,v,i-1)\vee\comp(u,v,i)$\`\\

\>{\color{white}***}\` for $uv\in E$, $u<v$, and $2\leq i\leq\omega$
\end{tabbing}
This completes the construction of the formula $F(G,\omega)$\lv{, which
has at most $\bigO(n^2)$ variables and at most $\bigO(n^3\omega)$
clauses}.
By construction, $F(G,\omega)$ is satisfiable
if and only $G$ has a derivation of length at most $\omega$. Because
of Theorem~\ref{the:equi-td-der}, we obtain:
\begin{theorem}
  The formula $F(G,\omega)$ is satisfiable if and only if $G$ has a treedepth
  at most $\omega-1$. Moreover,
  a corresponding treedepth decomposition can be constructed from a satisfying
  assignment of $F(G,\omega)$ in linear time in terms of the number of
  variables of $F(G,\omega)$. 
\end{theorem}

\lv{
\subsection{A Treedepth Encoding Based on Tree-Structure}
\label{asec:td-enc-2}

In~this section we introduce our second encoding for treedepth that is
based on guessing the tree-structure of a treedepth decomposition.
Let $G$ be a graph with $n$ vertices and $m$ edges, and let $\omega$
be a positive integer. As before we will assume that the
vertices and edges of $G$ are represented by numbers from $1$ to $n$
respectively $m$. Our aim is to construct
a formula $F(G,\omega)$ that is satisfiable if and only if $G$ has a
treedepth at most $\omega$. 
Informally, our encoding guesses a treedepth decomposition of $G$ by
guessing the roots as well as the parent relation of the underlying
forest. Namely, using a \emph{root variable} $\ro(r)$ for every $r \in
V(G)$ the encoding
guesses all the roots of the forest and using a \emph{parent
  variable} $\pt(p,c)$ for every $p,c \in V(G)$ with $p \neq c$ it
guesses the parent $p$ for every non-root vertex $c$. To ensure the
properties of a treedepth decomposition, the encoding then ``computes'' the
transitive closure of the relation $\pt(p,c)$. For this purpose the encoding
uses a variable $\anc(a,d)$ for every $a,d \in V(G)$ that can only be
true if $(a,d)$ is in the transitive closure of the relation
$\pt(p,c)$.
These variables can then be employed to verify the remaining
properties of a treedepth decomposition as follows. First to verify
that $G$ is a subgraph of the closure of the guessed treedepth
decomposition, it is sufficient to verify that $\anc(u,v)$ or
$\anc(v,u)$ is true for every edge $\{u,v\} \in E(G)$. Moreover, to verify
that the height of the guessed decomposition does not exceed $\omega$,
it suffices to check that the number of vertices $a$ for which
$\anc(a,c)$ holds, is at most $\omega-1$ for every $c \in V(G)$.


We start by introducing the clauses that together ensure that the
assignment of the root and parent variables corresponds to a treedepth
decomposition (i.e., an undirected forest).
Towards this aim we introduce the following clauses.
\begin{tabbing}
  \quad \quad \=  \kill
  \> (A) \quad $\bigvee_{r \in V(G)}\ro(r)$\\
  \> (B) \quad $\ro(r) \vee \big(\bigvee_{p \in V(G) \land p \neq
    r}\pt(p,r)\big)$\` for $r \in V(G)$\\
  \> (C) \quad $\neg \pt(p,c) \vee \neg \pt(p',c)$\` for $p,p',c \in V(G)$,\\
  \> {\color{white} (C)} \quad \` $p<p'$, $c\neq p$, $c\neq p'$\\
  \> (D) \quad $\neg \ro(r) \vee \neg \anc(a,r)$\` for $r,a \in V(G)$, $r\neq a$
\end{tabbing}
Note that (A) ensures that there is at least one root,
(B) ensures that every non-root vertex has at least one
parent, (C) ensures that every vertex has at most one parent, and (D)
ensures that a root vertex does not have any ancestors. This almost
ensures that the assignment of the root and parent variables
corresponds to a forest, since it only remains to ensure that the
directed graph with vertex set $V(G)$ having an arc from $c$ to $p$ if
$\pt(p,c)$ is true is acyclic. We achieve this by forcing that the transitive
closure of $\pt(p,c)$ (represented by $\anc(a,d)$) is irreflexive and
anti-symmetric.
\begin{tabbing}
  \quad \quad \=  \kill
  \> $\neg \anc(a,a)$\` for all $a \in V(G)$\\
  \> $\neg \anc(u,v) \vee \neg \anc(v,u)$ \` for all $u,v \in V(G)$,
  $u \neq v$ 
\end{tabbing}
\noindent For the above clauses to work, it is crucial to ensure that the
relation $\anc(a,d)$ is equal to the transitive closure of the
relation $\pt(p,c)$. Towards this aim, we first introduce the
following clauses ensuring that the transitive closure of $\pt(p,c)$
is contained in the relation $\anc(a,d)$.
\begin{tabbing}
  \quad \quad \=  \kill
  \> $\neg \pt(p,c)\vee \anc(p,c)$ \` for all $p,c\in V(G)$, $p\neq c$\\
  \> $\neg \pt(p,c)\vee \neg \anc(p',p)\vee\anc(p',c)$\` \\
  \> {\color{white}***} \quad \` for all $p',p,c
  \in V(G)$, $p'\neq p$, $p\neq c$, $p'\neq c$
\end{tabbing}
To ensure the converse, i.e., that the relation $\anc(a,d)$ is
contained in the transitive closure of the relation $\pt(p,c)$, we
introduce the following clauses.
\begin{tabbing}
  \quad \quad \=  \kill
  \> $\neg \pt(p,c)\vee \neg \anc(p',c)\vee \anc(p',p)$\` \\
    \> {\color{white}***} \quad \` for all $p',p,c
  \in V(G)$, $p'\neq p$, $p\neq c$, $p'\neq c$
\end{tabbing}
\noindent By now it only remains to ensure that (1) the height of the
forest is at most $\omega$ and (2) every edge of $G$ is in the closure
of the forest. The former is
achieved by employing the sequential counter introduced in
Section~\ref{ssec:pre-formula} to force that $|\SB a \SM
\anc(a,l)\SE|$ is at most $\omega-1$ for every $l \in
V(G)$ and
the later is achieved using the following clauses.
\begin{tabbing}
  \quad \quad \=  \kill
  \> $\anc(u,v)\vee \anc(u,v)$\` for all $\{u,v\} \in E(G)$
\end{tabbing}

\sloppypar
This completes the construction of the formula $F(G,\omega)$, which
has at most $\bigO(n^2\omega)$ variables and at most $\bigO(n^3)$
clauses.
By construction, we obtain the following theorem.
\begin{theorem}\label{the:formula-tc}
  The formula $F(G,\omega)$ is satisfiable if and only if $G$ has a
  treedepth at most $\omega$. Moreover,
  a corresponding treedepth decomposition can be constructed from a satisfying
  assignment of $F(G,\omega)$ in linear time in terms of the number of
  variables of $F(G,\omega)$. 
\end{theorem}
}

\subsection{Preprocessing and Symmetry Breaking}

To increase the efficiency of our encoding, we implemented a
number of preprocessing procedures and symmetry breaking rules.
Our first symmetry breaking rule is based on the
next lemma.
\lv{\begin{lemma}}
  \sv{\begin{lemma}[$\star$]}
  \label{lem:td-twins}
  Let $G$ be a graph and let $u$ and $v$ be two adjacent vertices in
  $G$ such that $N_G(u)\setminus \{v\} \subseteq N_G(v) \setminus
  \{u\}$. Then there is an optimal treedepth decomposition $F$ of $G$ such
  that $v$ is an ancestor of $u$ in $F$.
\end{lemma}
\lv{\begin{proof}
  Because $u$ and $v$ are adjacent, it follows that either $u$ is an
  ancestor of $v$ or $v$ is an ancestor of $u$ in any treedepth
  decomposition of $G$. Hence let $F$ be a treedepth decomposition of
  $G$ such that $u$ is an ancestor of $v$. We claim that the forest
  $F'$ obtained from $F$ by switching $u$ and $v$ is still a treedepth
  decomposition of $G$. Indeed, $N_{\undC(F)}(v) \subseteq
  N_{\undC(F')}(v)$ and hence all edges incident to $v$ in $G$ are
  still covered by the closure of $F'$. Moreover, since $N_G(u) \setminus \{v\}
  \subseteq N_G(v) \setminus \{u\} \subseteq
  N_{\undC(F)}(v)\setminus \{u\}=N_{\undC(F')}(u)\setminus \{v\}$ it follows that all edges incident to $u$ in
  $G$ are still covered by the closure of $F'$.
  \qed
\end{proof}}
To employ the above lemma in our encoding, we iterate over all edges
of $G$ and whenever we find an edge $\{u,v\} \in E(G)$ such that
$N_G(u)\setminus \{v\}\subseteq N_G(v)\setminus \{u\}$, we add the
clause $\neg \comp(u,u,i) \vee \comp(v,v,i)$ for every
$i$ with $2 \leq i \leq \omega$.
\sv{
  We also introduce two preprocessing procedures, whose correctness is
  shown the lemma below. One allows us to remove certain vertices of
  degree one and the other allows us to remove apex
  vertices, i.e., vertices (whose) closed neighborhood is the whole vertex set.
%
  \begin{lemma}[$\star$]\label{lem:td-pre-comb}
      Let $G$ be a graph. If $v$ is a vertex of $G$ incident to two
      vertices $l$ and $l'$ of degree one, then $\td(G)=\td(G-\{l'\})$. Moreover, if $v$ is an apex vertex of $G$, then
      $\td(G)=\td(G-\{v\})+1$.
  \end{lemma}
}

The following lemma allows us to remove certain vertices of degree $1$ from the graph.
  \begin{lemma}
    \label{lem:td-leaves}
  Let $G$ be a graph and let $v$ be a vertex of $G$ incident to two
  vertices $l$ and $l'$ of degree one. Then $\td(G)=\td(G-\{l'\})$.
\end{lemma}
\begin{proof}
  Because $G \setminus \{l'\}$ is a subgraph of $G$, we obtain that
  $\td(G)\geq \td(G- \{l'\})$. Towards showing that
  $\td(G)\leq \td(G- \{l'\})$, let $F$ be an optimal
  treedepth decomposition of $G \setminus \{l'\}$. Because of
  Lemma~\ref{lem:td-twins}, we can assume that $v$ is an
  ancestor of $l$ in $F$. Consequently, the forest $F'$ obtained from
  $F$ after adding $l'$ as a leaf to $v$ has the same depth as $F$ and
  is a treedepth decomposition of $G$.
  \qed
\end{proof}

Our final lemma allows us to remove all apex
  vertices.
\begin{lemma}
    \label{lem:td-apex}
  Let $G$ be a graph and let $a$ be an apex vertex of $G$. Then
  $\td(G)=1+\td(G \setminus \{a\})$.
\end{lemma}
\begin{proof}
  Towards showing that $\td(G) \leq 1+\td(G \setminus \{a\})$, let $F$
  be an optimal treedepth decomposition of $G \setminus \{a\}$. Then,
  $F'$, which is obtained from $F$ by simply adding $a$, making it
  adjacent to all roots of $F$, and setting it to be the new root of
  $F'$, is a treedepth decomposition of $G$ of width at most
  $1+\td(G-\{a\})$, as required.

  Towards showing that
  $\td(G) \geq 1+\td(G- \{a\})$, let $F$ be an optimal
  treedepth decomposition of $G$. By
  Lemma~\ref{lem:td-twins}, we can assume that $a$ is a root of
  $F$. Moreover, because $a$ is adjacent to every vertex in $G$, it
  follows that $a$ must be the only root of $F$. Hence $F'$ obtained
  from $F \setminus \{a\}$ after making all children of $a$ in $F$ to
  roots is a treedepth decomposition of $G- \{a\}$ of depth at
  most $\td(G)-1$, as required.
\end{proof}

\section{Experiments}\label{sec:experiments}

We implemented the SAT-encoding for treecut width and the two SAT-encodings for treedepth and evaluated them on various benchmark
instances; for comparison we also computed the pathwidth and
treewidth of all graphs using the currently best performing
SAT-encodings~\cite{SamerVeith09,LodhaOrdyniakSzeider16}; note
that~\cite{SamerVeith09} is still the best-known SAT-encoding for
treewidth, since the performance gains of later algorithms~\cite{BergJarvisalo14,BannachBerndtEhlers17} are
almost entirely due to preprocessing, whereas the employed SAT-encoding is virtually identical.
Our benchmark instances include $39$ famous named graphs
from the literature~\cite{MathWorld}, various standard
graphs such as complete graphs ($K_n$), complete bipartite graphs
($K_{n,n}$), paths ($P_n$), cycles ($C_n$), complete binary trees
($B_n$), and grids ($G_{n,n}$) as well as random
graphs. To test the correctness of our
encodings, we compared the obtained values to known values (the treedepth of standard graphs is known or
easy to compute; in the case of treecut width, the values for complete
and complete bipartite graphs are given in
\lv{Lemmas~\ref{lem:tcw-comp} and~\ref{lem:tcw-comp-bip}}\sv{Lemma~\ref{lem:tcw-comp-sv}}). We also compared the
obtained values to related parameters such as maximum degree, pathwidth, and treewidth
(using Lemmas~\ref{lem:tcw-rel} and~\ref{pro:tdfacts}) and
verified that the decompositions obtained from the encodings are well-formed.


Throughout we used the SAT-solver Glucose 4.0 (with standard parameter settings) 
as it performed best in our initial tests. We ran the experiments on a 
4-core Intel Xeon CPU E5649, 2.35GHz, 72 GB RAM machine with Ubuntu 14.04 with 
each process having access to at most 8 GB RAM.
Our implementation,
 which was done in python~2.7.3 and networkx~1.11, is available via 
 \texttt{GitHub}~\footnote{\url{
 https://github.com/nehal73/TCW_TD_to_SAT}}. 

\subsection{Results and Discussion}\label{ssec:results}

\begin{table}[tbh]
\centering{
\caption{Experimental results for standard graphs.
A ``$P$'' indicates that the
instance is solved by preprocessing.
}\label{tab:constructed}
\vspace{5pt}

\begin{tabular}{@{}l@{~~~}
r@{~~}r@{~~~~~}r@{~~~}r@{}}
\toprule
\multirow{2}{*}{\emph{graph class}} & \multicolumn{2}{c}{\emph{treecut width}} & 
\multicolumn{2}{c}{\emph{treedepth}} \\ \cmidrule(r){2-3}\cmidrule(r){4-5} 
 & {$|V|$} & {$|E|$} & 
                       {$|V|$} & {$|E|$} \\ \midrule
  paths ($P_n$) & $P$ & $P$ & 255 & 254 \\
  cycles ($C_n$) &$P$  & $P$ & 255 & 255 \\
  complete binary trees ($B_n$) & $P$ & $P$ & 255 & 254 \\
  $n\times n$ grids ($G_{n,n}$) & 49 & 84 & 36 & 60 \\
  complete bip. graphs ($K_{n,n}$) & 30 & 225 & 22 & 121 \\
  complete graphs ($K_n$) & 30 & 435 & $P$ & $P$ \\ \bottomrule
\end{tabular}
}
\vspace{-0.1cm}
\end{table}
\begin{table}[tbh]
\centering\caption{
  Percentage of random graphs solved within the timeout for
  treecut width and treedepth  for  combinations
  of $|V|$ (represented by the rows) and $p$ (edge probability; 
  represented by the columns).}

\vspace{5pt}

\label{tab:exp-rand}
\begin{tabular}{@{}l@{~~~~}r@{~~~}r@{~~~}r@{~~~}r@{~~~}r@{~~~}r@{~~~}
r@ {~~~} r@{~~~}r@{}}
\toprule
\multicolumn{10}{c}{\emph{treecut width}}\\
\midrule
  $|V|$ & 0.1 & 0.2 & 0.3 & 0.4 & 0.5 & 0.6 & 0.7 & 0.8 & 0.9 \\ \midrule
  20 & 100 & 100 & 100 & 100 & 100 & 100 & 100 & 100 & 100 \\
  25 & 100 & 100 & 75  & 90  & 100 & 100 & 100 & 100 & 100 \\
  30 & 100 & 25  & 10  & 55  & 85  & 100 & 100 & 100 & 100 \\
  40 & 50  & 0   & 0   & 0   & 0   & 0   & 0   & 0   & 0   \\
  50 & 0   & 0   & 0   & 0   & 0   & 0   & 0   & 0   & 0   \\
  \midrule \\
\toprule
\multicolumn{10}{c}{\emph{treedepth} (partition-based encoding)}\\
\midrule
  $|V|$ & 0.1 & 0.2 & 0.3 & 0.4 & 0.5 & 0.6 & 0.7 & 0.8 & 0.9 \\ \midrule
  10 & 100 & 100 & 100 & 100 & 100 & 100 & 100 & 100 & 100 \\
  15 & 100 & 100 & 100 & 100 & 100 & 100 & 100 & 100 & 100 \\
  20 & 100 & 100 & 100 & 45  & 10  & 0   & 0   & 0   & 15  \\
  30 & 100 & 0   & 0   & 0   & 0   & 0   & 0   & 0   & 0   \\
  40 & 10  & 0   & 0   & 0   & 0   & 0   & 0   & 0   & 0   \\
  50 & 0   & 0   & 0   & 0   & 0   & 0   & 0   & 0   & 0 \\
  \midrule \\
\toprule
\multicolumn{10}{c}{\emph{treedepth} (tree-structure based encoding)}\\
\midrule
  $|V|$ & 0.1 & 0.2 & 0.3 & 0.4 & 0.5 & 0.6 & 0.7 & 0.8 & 0.9 \\ \midrule
  10 & 100 & 100 & 100 & 100 & 100 & 100 & 100 & 100 & 100 \\
  15 & 100 & 85  & 35  & 10  & 0   & 5   & 10  & 50  & 100 \\
  20 & 75  & 5   & 0   & 0   & 0   & 0   & 0   & 0   & 30  \\
  25 & 25  & 0   & 0   & 0   & 0   & 0   & 0   & 0   & 0   \\
  30 & 0   & 0   & 0   & 0   & 0   & 0   & 0   & 0   & 0   \\ \bottomrule
\end{tabular}
\end{table}

\begin{table*}[p]
\centering{
  \caption{Experimental results for the famous graphs.
  ``cpu'' denotes the overall CPU time in
seconds including preprocessing and verification of the computed decomposition. 
An asterisk ($*$) in the cpu column indicates that the given instance could not be solved within the timeout; in this case the width
column gives the lower bound and upper bound obtained within the
timeout. }\label{tab:famous}

\vspace{10pt}
\begin{tabular}{@{}l@{~~~~~~~~~}r@{~~~~~~~}r@{~~~~~}r@{~~~~~}r@{~~~~~}r@{
~~~~~~~~ } r@ { ~~~~~ } r@ {~~~~~~~~}r@{~~~~~}r@{}}
\toprule
\multirow{2}{*}{Instance} & \multirow{2}{*}{$|V|$} & \multirow{2}{*}{$|E|$} & 
\multirow{2}{*}{$\Delta$} & \multicolumn{2}{c}{treecut width} 
& \multicolumn{2}{c}{treedepth} & \multirow{2}{*}{$\pw$} & 
\multirow{2}{*}{$\tw$} \\ \cmidrule(lr){5-6}\cmidrule(r){7-8}
 &  &  &  & \multicolumn{1}{c}{width}  & \multicolumn{1}{c}{cpu} & 
\multicolumn{1}{c}{width} & \multicolumn{1}{c}{cpu} &  &  \\ \midrule
Diamond & 4 & 5 & 3 & 2 & 0.00 & 3 & 0.15 & 2 & 2 \\
Bull & 5 & 5 & 3 & 2 & 0.00 & 3 & 0.07 & 2 & 2 \\
Butterfly & 5 & 6 & 4 & 2 & 0.00 & 3 & 0.06 & 2 & 2 \\
Prism & 6 & 9 & 3 & 4 & 0.13 & 5 & 0.10 & 3 & 3 \\
Moser & 7 & 11 & 4 & 4 & 0.12 & 5 & 0.15 & 3 & 3 \\
Wagner & 8 & 12 & 3 & 4 & 0.21 & 6 & 0.26 & 4 & 4 \\
Pmin & 9 & 12 & 3 & 4 & 0.13 & 5 & 0.14 & 4 & 3 \\
Petersen & 10 & 15 & 3 & 5 & 0.71 & 6 & 0.34 & 5 & 4 \\
Herschel & 11 & 18 & 4 & 5 & 0.86 & 5 & 0.13 & 4 & 3 \\
Gr\"{o}tzsch & 11 & 20 & 5 & 6 & 1.07 & 7 & 0.29 & 5 & 5 \\
Goldner & 11 & 27 & 8 & 7 & 2.12 & 5 & 0.25 & 4 & 3 \\
D\"{u}rer & 12 & 18 & 3 & 4 & 0.85 & 7 & 0.37 & 4 & 4 \\
Franklin & 12 & 18 & 3 & 4 & 0.71 & 7 & 0.34 & 5 & 4 \\
Frucht & 12 & 18 & 3 & 4 & 0.83 & 6 & 0.23 & 4 & 3 \\
Tietze & 12 & 18 & 3 & 5 & 1.20 & 7 & 0.39 & 5 & 4 \\
Chv\'{a}tal & 12 & 24 & 4 & 6 & 1.85 & 8 & 0.68 & 6 & 6 \\
Paley13 & 13 & 39 & 6 & 10 & 6.31 & 10 & 4.60 & 8 & 8 \\
Poussin & 15 & 39 & 6 & 9 & 22.36 & 9 & 2.64 & 6 & 6 \\
Sousselier & 16 & 27 & 5 & 6 & 6.31 & 8 & 1.20 & 5 & 5 \\
Hoffman & 16 & 32 & 8 & 4 & 8.83 & 8 & 1.74 & 7 & 6 \\
Clebsch & 16 & 40 & 5 & 8 & 7.68 & 10 & 18.41 & 9 & 8 \\
Shrikhande & 16 & 48 & 6 & 10 & 18.86 & 11 & 49.77 & 9 & 7-10 \\
Errera & 17 & 45 & 6 & 9 & 17.84 & 10 & 19.93 & 6 & 6 \\
Paley17 & 17 & 68 & 6 & 14 & 51.20 & 14 & 7569.02 & 12 & 11 \\
Pappus & 18 & 27 & 3 & 6 & 35.26 & 8 & 1.92 & 7 & 5-7 \\
Robertson & 19 & 38 & 4 & 8 & 42.64 & 10 & 63.01 & 8 & 7-9 \\
Desargues & 20 & 30 & 3 & 6 & 56.18 & 9 & 12.36 & 7 & 5-7 \\
Dodecahedron & 20 & 30 & 3 & 6 & 87.05 & 9 & 10.84 & 6 & 5-7 \\
FlowerSnark & 20 & 30 & 3 & 6 & 76.37 & 9 & 17.45 & 7 & 5-7 \\
Folkman & 20 & 40 & 4 & 8 & 78.64 & 9 & 11.77 & 7 & 6 \\
Brinkmann & 21 & 42 & 4 & 8 & 75.38 & 11 & 838.41 & 8 & 7-10 \\
Kittell & 23 & 63 & 7 & 10 & 65.11 & 12 & 2422.53 & 7 & 7 \\
McGee & 24 & 36 & 3 & 6 & 71.78 & 11 & 2825.19 & 8 & 5-8 \\
Nauru & 24 & 36 & 3 & 6 & 52.68 & 10 & 158.20 & 8 & 5-8 \\
Holt & 27 & 54 & 4 & [7-9] & * & [11-13] & * & 9 & 7-10 \\
Watsin & 50 & 75 & 3 & 5 & 202.49 & [10-13] & * & 7 & 5-8 \\
B10Cage & 70 & 105 & 3 & [5-11] & * & [10-23] & * & [11-16] & 5-17 \\
Ellingham & 78 & 117 & 3 & 6 & 15002.47 & [10-14] & * & 6 & 5-7 \\ \bottomrule
\end{tabular}}
\end{table*}
 
Our experimental results for the standard, random, and famous graphs
are shown in Tables~\ref{tab:constructed},~\ref{tab:exp-rand}, and~\ref{tab:famous},
respectively. Throughout we use $|V|$, $|E|$, $\Delta$, $\pw$, $\tw$
to denote the number of vertices, the number of edges, the maximum
degree, the pathwidth, and the treewidth of the input graph,
respectively. 
We employed a timeout per SAT call of $2000$ seconds and an
overall timeout of $6$ hours for our experiments with the famous and random
graphs. Moreover, we used $900$ seconds per SAT call and an overall
timeout of $3$ hours for the standard graphs.

As can be seen in Table~\ref{tab:famous}, we were able to compute the
exact treecut width and treedepth for almost all of the famous graphs;
specifically $37$ out of $39$ instances for treecut width and $35$ out of $39$
instances for treedepth. For the remaining two respectively four
instances, we were able to obtain relatively tight lower and
upper bounds. Even though we are aware that encodings for different width measures
are not directly comparable, it is interesting to note that our encodings outperform the currently best
performing SAT-encoding for treewidth~\cite{SamerVeith09}, which solves only $26$ out of
$39$ instances, and are in line with the currently
best performing SAT-encoding for
pathwidth~\cite{LodhaOrdyniakSzeider17}, solving $37$ out of $39$
instances.
It is also worth mentioning that the surprisingly strong performance
of our encodings is not due to preprocessing; indeed, none of the
preprocessing or symmetry-breaking rules for treecut width nor
treedepth were applicable for the famous graphs.
Finally, we would like to mention that our second encoding for treedepth could only solve $17$ out of $39$ instances. This further
underlines the strength of partition-based encodings for
computing decomposition-based parameters.

Table~\ref{tab:constructed} shows the scalability of our encodings for
the standard graphs. Namely, for each of the standard graphs and both
of our encodings, the table gives the maximum number of vertices and
edges for which the encoding was able to compute the exact width
within the timeout. Note that not all of the standard graphs are
interesting for both treecut width and treedepth (indicated by a ``$P$''
in the table). This is because some of the graphs can be solved entirely by
preprocessing; for instance the treedepth of complete graphs can be
computed using Lemma~\lv{\ref{lem:td-apex}}\sv{\ref{lem:td-pre-comb}}. Moreover, the treecut
width of paths, cycles, and trees can be computed using
Lemma~\ref{lem:tcw-3ec}. As one can see, our treecut width encoding is
able to solve $K_n$, $K_{n,n}$, and $G_{n,n}$ up to $n=30$, $n=15$,
and $n=7$, respectively. Similarly, our treedepth encoding is able to
solve $B_n$, $K_{n,n}$, $C_n$, $G_{n,n}$, and $P_n$ up to $n=8$,
$n=11$, $n=255$, $n=6$, and $n=255$, respectively. Given the
simplicity of the treedepth encoding it is surprising that it
performs slightly worse than the treecut width encoding on complete
bipartite graphs and grids. However, its extraordinary performance on
paths, complete binary trees and cycles seems to indicate that the
encoding is well suited for sparse graphs.
 
Finally, Table~\ref{tab:exp-rand} provides the scalability of our
three encodings for uniformally generated random graphs for varying edge
densities and number of vertices. In line with our previous
observations the encoding for treecut width scales significantly
better than our two encodings for treedepth (solving almost all random
graphs upto $30$ instead of $15$ vertices) and both encodings show a
slight preference for very sparse graphs. For the case of treedepth,
the results once again show a significant advantage for our partition-based
encoding over the tree-structure based encoding.

\section{Conclusion and Future Work}

We implemented the first practical algorithms for computing the
algorithmically important parameters treecut width and treedepth. Our
experimental results show that due to our novel partition-based characterisations for the considered width parameters, our algorithms perform very well on
small to medium sized graphs. In particular, our algorithms perform
better than the current best SAT-encoding for treewidth, which even
though not directly comparable serves as a good reference point. We would also like to point out that our algorithms will be
very helpful in the future to evaluate the accuracy of heuristics for
the considered decomposition parameters and can be scaled to larger
graphs if the aim is just to compute lower bounds and upper bounds for
the parameters. We see our algorithms as
a first step towards turning the yet mostly theoretical applications of both
parameters into practice.

Extending the scalability of our
algorithms to even larger graphs can be seen as the main challenge
for future work. Here, SAT-based local improvement approaches 
such as those that have recently been developed for branchwidth and treewidth~\cite{LodhaOrdyniakSzeider16,FichteLodhaSzeider17},
provide an interesting venue for future work.
In fact, the work on local improvement for treewidth~\cite{FichteLodhaSzeider17} showed that, compared to other exact methods, SAT-encodings are particularly suitable for this approach, hence it can be expected that our SAT-encodings for treedepth and treecut width will serve well in a local improvement approach. 
Other promising directions include the development of
more efficient preprocessing procedures, or
splitting the
graph into smaller parts by using, e.g.,
balanced cuts or separators.


\smallskip
\noindent \textbf{Errata and acknowledgments.} The short version of this article which appeared at ALENEX 2019 contained $2$ erroneous treedepth values in Table~\ref{tab:famous}. This was caused by an incorrect transition from the preprocessing to the solver: in particular, the solver requires the graph to be connected, and hence it is necessary to provide it with the individual connected components that arise from preprocessing. The issue is fixed in the presented version.

We thank James Trimble (of the School of Computing Science at the University of Glasgow) for spotting this issue. We also thank Vaidyanathan P. R. (of the Algorithms and Complexity Group at TU Wien) for his help with resolving the transition issue described above.

\bibliographystyle{plain}
\bibliography{literature}

\end{document}